\documentclass[10pt]{llncs}
\usepackage{enumerate,amsmath,amssymb,mathrsfs}
\usepackage{graphicx}
\usepackage{color}
\usepackage[all]{xy}

\renewcommand{\phi}{\varphi}

\newcommand{\ignore}[1]{}

\begin{document}

\title{Switchability and collapsibility of Gap Algebras}
\author{Dmitriy Zhuk\inst{1} and Barnaby Martin\inst{2}\thanks{The second author was supported by EPSRC grant EP/L005654/1.}}
\institute{
  Moscow State University, 119899 Moscow, Russia. \\
  \url{zhuk@intsys.msu.ru}\\
  \and
  School of Science and Technology, Middlesex University, \\
  The Burroughs, Hendon, London, NW4 4BT, UK.\\
  \url{barnabymartin@gmail.com}\\
}
\maketitle

\begin{abstract}
Let $\mathbb{A}$ be an idempotent algebra on a $3$-element domain $D$ that omits a $G$-set for a factor. Suppose  $\mathbb{A}$ is not $\alpha\beta$-projective (for some $\alpha,\beta \subset D$ so that $\alpha \cup \beta = D$ and $\alpha \cap \beta \neq \emptyset$) and is not collapsible. It follows that $\mathbb{A}$ is switchable. We prove that, for every finite subset $\Delta$ of Inv$(\mathbb{A})$, Pol$(\Delta)$ is collapsible.
We also exhibit an algebra that is collapsible from a non-singleton source but is not collapsible from any singleton source.
\end{abstract}

\section{Introduction}

For a finite-domain algebra $\mathbb{A}$ we associate a function
$f_\mathbb{A}:\mathbb{N}\rightarrow\mathbb{N}$, giving the cardinality
of the minimal generating sets of the sequence $\mathbb{A},
\mathbb{A}^2, \mathbb{A}^3, \ldots$ as $f(1), f(2), f(3), \ldots$,
respectively. We may say $\mathbb{A}$ has the $g$-GP if $f(m) \leq
g(m)$ for all $m$. The question then arises as to the growth rate of
$f$ and specifically regarding the behaviours constant, logarithmic,
linear, polynomial and exponential. Wiegold proved in
\cite{WiegoldSemigroups} that if $\mathbb{A}$ is a finite semigroup
then $f_{\mathbb{A}}$ is either linear or exponential, with the former
prevailing precisely when $\mathbb{A}$ is a monoid. This dichotomy
classification may be seen as a gap theorem because no growth rates
intermediate between linear and exponential may occur. We say
$\mathbb{A}$  enjoys the \emph{polynomially generated powers} property
(PGP) if there exists a polynomial $p$ so that $f_{\mathbb{A}}=O(p)$
and  the \emph{exponentially generated powers} property (EGP) if there
exists a constant $b$ so that $f_{\mathbb{A}}=\Omega(g)$ where
$g(i)=b^i$.

A great literature of work exists from the past twenty years on
applications of universal algebra in the computational complexity of
\emph{constraint satisfaction problems} (CSPs) and a number of
celebrated results have been obtained through this method. 
Each
CSP is parameterised by a finite structure $\mathcal{B}$ and asks whether an input
sentence $\varphi$ holds on $\mathcal{B}$, where $\varphi$ is a
primitive positive sentence, that is where only $\exists$ and $\land$ may be
used. For almost every class of model checking problem induced by the presence
or absence of first-order quantifiers and connectors, we can give a complexity
classification~\cite{DBLP:journals/corr/abs-1210-6893}: the two
outstanding classes are CSPs and its popular extension \textsl{quantified
  CSPs} (QCSPs) for positive Horn sentences -- where $\forall$ is also present -- which is used in
Artificial Intelligence to model non-monotone reasoning or
uncertainty. 

In Hubie Chen's \cite{AU-Chen-PGP}, a new link between algebra and
QCSP was discovered. Chen's previous work in QCSP tractability largely
involved the special notion of \emph{collapsibility}
\cite{hubie-sicomp}, but in \cite{AU-Chen-PGP} this was extended to a \emph{computationally effective} version of the PGP.

For a finite-domain, idempotent algebra $\mathbb{A}$, \emph{$k$-collapsibility} may be seen as a special form of the PGP in which the generating set for $\mathbb{A}^m$ is constituted of all tuples $(x_1,\ldots,x_m)$ in which at least $m-k$ of these elements are equal. \emph{$k$-switchability} may be seen as another special form of the PGP in which the generating set for $\mathbb{A}^m$ is constituted of all tuples $(x_1,\ldots,x_m)$ in which there exists $a_i<\ldots<a_{k'}$, for $k'\leq k$, so that
\[ (x_1,\ldots,x_m) = (x_1,\ldots,x_{a_1},x_{a_1+1},\ldots,x_{a_2},x_{a_2+1},\ldots,\ldots,x_{a_k'},x_{a_k'+1},\ldots,x_m),\]
where $x_1=\ldots=x_{a_1-1}$, $x_{a_1}=\ldots=x_{a_2-1}$, \ldots, $x_{a_{k'}}=\ldots=x_{a_m}$. Thus, $a_1,a_2,\ldots,a_{k'}$ are the indices where the tuple switches value. Note that these are not the original definitions but they are proved equivalent to the original definitions in \cite{LICS2015}. We say that $\mathbb{A}$ is collapsible (switchable) if there exists $k$ such that it is $k$-collapsible ($k$-switchable). For any finite algebra, $k$-collapsibility implies $k$-switchability and for any $2$-element algebra, $k$-switchability implies $k$-collapsibility. Chen originally introduced switchability because he found a $3$-element algebra that enjoyed the PGP but was not collapsible \cite{AU-Chen-PGP}. He went on to prove that switchability of $\mathbb{A}$ implies that the corresponding QCSP is in P, what one might informally state as QCSP$(\mathrm{Inv}(\mathbb{A}))$ in P, where $\mathrm{Inv}(\mathbb{A})$ can be seen as the structure over the same domain as $\mathbb{A}$ whose relations are precisely those that are preserved by (invariant under) all the operations of $\mathbb{A}$. However, the QCSP is typically defined only on finite sets of relations (else the question arises as to encoding), thus a more formal definition might be that, for any finite subset $\Delta$ of $\mathrm{Inv}(\mathbb{A})$, QCSP$(\Delta)$ is in P. What we prove in this paper is that, as far as the QCSP is concerned, switchability on a $3$-element algebra is something of a mirage. What we mean by this is that we pursue Chen's line of thinking from \cite{AU-Chen-PGP} and consider $3$-element algebras $\mathbb{A}$ that omit a $G$-set as a factor\footnote{We will not define what this means and will not comment much further on it since it was an assumption chosen by Chen in \cite{AU-Chen-PGP}. Note that, if idempotent $\mathbb{A}$ has a $G$-set as a factor, then  QCSP$(\mathrm{Inv}(\mathbb{A}))$ is NP-hard.} and do not have the EGP: he proves such algebras are switchable. What we prove is that for these same algebras, for any finite subset $\Delta$ of Inv$(\mathbb{A})$, Pol$(\Delta)$ is actually collapsible. Thus, for QCSP complexity here, we do not need the additional notion of switchability to explain tractability, as collapsibility will already suffice. Since these notions were introduced in connection with the QCSP this is particularly surprising. Note that the parameter $k$ of collapsibility is unbounded over these increasing finite subsets $\Delta$ while the parameter of switchability clearly remains bounded. In some way we are suggesting that switchability itself might be seen as a limit phenomenon of collapsibility. This is hard to see on the algebraic side but both collapsibiliy and switchability admit equivalent definitions on the relational side where this is more apparent. 

Zhuk settled the PGP versus EGP dichotomy for finite-domain algebras in \cite{ZhukGap2015} (indeed his dichotomy holds even in the non-idempotent case). All of the PGP cases are switchable.

Finally, at the end of the paper we exhibit an algebra over $3$ elements that is collapsible from a non-singleton source but is not collapsible from any singleton source. This answers a question left open from the work \cite{LICS2015} (though not specifically mentioned there).

\section{Preliminaries}

In this paper we will use the original definition of collapsibility, as given in  \cite{hubie-sicomp}, rather than the equivalent definition in the introduction. An $m$-ary \emph{adversary} over a finite set $D$ is an $m$-tuple of subsets of $D$. Suppose $f$ is an idempotent $k$-ary function on $D$. We say the adversary $(A_1,\ldots,A_m)$ is \emph{$f$-composable} from adversaries $(A^1_1,\ldots,A^1_m)$, \ldots, $(A^1_1,\ldots,A^1_m)$ if we have $f(A^1_1,\ldots,A^k_1) \supseteq A_1$, \ldots,  $f(A^1_m,\ldots,A^k_m) \supseteq A_m$. Let $\mathbb{A}$ be an idempotent algebra on finite $D$. Define a special set of adversaries $\Sigma^D_{m,k,x}$, for each $x \in D$, which contains all co-ordinate permutations of the adversaries $(D,\ldots,D,\{x\},\ldots,\{x\})$ where $D$ appears $k$ times and $\{x\}$ appears $m-k$ times. $\mathbb{A}$ is \emph{$k$-collapsible from source $X\subseteq D$} \cite{hubie-sicomp} if, for all $m$, there exists a term $f$ of $\mathbb{A}$ so that the adversary $(D,\ldots,D)$ is $f$-composable from the adversaries $\bigcup_{x \in X} \Sigma^D_{m,k,x}$. When the source $X$ is the whole set $D$, we just talk about being $k$-collapsible. We further say that $\mathbb{A}$ is \emph{collapsible} if it is $k$-collapsible for some $k$. Again we remind that this is one of a number of equivalent definitions of collapsibility proved equivalent in \cite{LICS2015}. We will not define switchability beyond what we mentioned in the introduction since, although it plays a role in this paper, this role is essentially non-technical. The only important point is that there is an algebra on a $3$-element domain that is switchable but not collapsible, that we will meet shortly. We note also that the original definition of switchability involves the so-called reactive composition of adversaries and is somewhat more complicated that that suggested in the introduction.

Let $f$ be a $k$-ary idempotent operation on domain $D$. We say $f$ is a \emph{generalised Hubie-pol} on $z_1\ldots z_k$ if, for each $i \in k$, $f(D,\ldots,D,z_i,D,\ldots,D)=D$ ($z_i$ in the $i$th position). When $z_1=\ldots=z_k=a$ this is called a \emph{Hubie-pol} in $\{a\}$ and gives $(k-1)$-collapsibility from source $\{a\}$. In general, a generalised Hubie-pol does not bestow collapsibility (\mbox{e.g.} Chen's $4$-ary switchable operation $r$, below).  The name Hubie operation was used in \cite{LICS2015} for Hubie-pol and the fact that this leads to collapsibility is noted in \cite{hubie-sicomp}.

\

\noindent \textbf{Globally}: let $\mathbb{A}$ be an idempotent algebra on a $3$-element domain $\{a,b,c\}:=D$. Assume $\mathbb{A}$ has precisely two subalgebras on domains $\{a,c\}$ and $\{b,c\}$ and contains the idempotent semilattice-without-unit operation $s$ which maps all tuples off the diagonal to $c$. Thus, $\mathbb{A}$ is a \emph{Gap Algebra} as defined in \cite{AU-Chen-PGP}. Note that the presence of $s$ removes the possibility to have a $G$-set as a factor. We say that $\mathbb{A}$ is $\{a,c\}\{b,c\}$-projective if for each $k$-ary $f$ in $\mathbb{A}$ there exists $i \leq k$ so that, if $x_i \in \{a,c\}$ then $f(x_1,\ldots,x_k) \in \{a,c\}$ and if $x_i \in \{b,c\}$ then $f(x_1,\ldots,x_k) \in \{b,c\}$. Let us now further assume that $\mathbb{A}$ is not $\{a,c\}\{b,c\}$-projective. This rules out the Gap Algebras that have EGP and we now know that $\mathbb{A}$ is switchable \cite{AU-Chen-PGP}. We will now consider the $4$-ary operation $r$ defined by Chen in \cite{AU-Chen-PGP}. Let $r$ be the idempotent operation satisfying
\[
\begin{array}{ccc}
abbb & & b \\
babb & r & b \\
aaab &\mapsto & a \\
aaba & & a \\
\mbox{else} & & c.
\end{array} 
\]
Chen proved that $(D;r,s)$ is $2$-switchable but not $k$-collapsible, for any $k$ \cite{AU-Chen-PGP}. Let $f$ be a $k$-ary operation in $\mathbb{A}$ that is not $\{a,c\}\{b,c\}$-projective. Violation of $\{a,c\}\{b,c\}$-projectivity in $f$ means that for each $i \in [k]$ either 
\begin{itemize}
\item there is $x_i \in \{a,b\}$ and $x_1,\ldots,x_{i-1},x_{i+1},\ldots,x_k \in \{a,b,c\}$ so that $f(x_1,\ldots,x_k)=y \in (\{a,b\}\setminus \{x_i\})$, or
\item or $x_i=c$ and there is $x_1,\ldots,x_{i-1},x_{i+1},\ldots,x_k \in \{a,b,c\}$ so that $f(x_1,\ldots,x_k)=y \in \{a,b\}$.
\end{itemize}
Note that we can rule out the latter possibility and further assume $x_1,\ldots,x_{i-1},$ $x_{i+1},\ldots,x_k \in \{a,b\}$, by replacing $f$ if necessary by the $2k$-ary $f(s(x_1,x'_1),\ldots,$ $s(x_k,x'_k))$. Thus, we may assume that (*) for each $i \in [k]$ there is $x_i \in \{a,b\}$ and $x_1,\ldots,x_{i-1},x_{i+1},\ldots,x_k \in \{a,b\}$ so that $f(x_1,\ldots,x_k)=y \in (\{a,b\}\setminus \{x_i\})$.

We wish to partition the $k$ co-ordinates of $f$ into those for which violation of $\{a,c\}\{b,c\}$-projectivity, on words in $\{a,b\}^k$:
\begin{itemize}
\item[$(i)$] happens with $a$ to $b$ but never $b$ to $a$.
\item[$(ii)$] happens with $b$ to $a$ but never $a$ to $b$.
\item[$(iii)$] happens on both $a$ to $b$ and $b$ to $a$.
\end{itemize}
Note that Classes $(i)$ and $(ii)$ are both non-empty (Class $(iii)$ can be empty). This is because if Class $(i)$ were empty then  $f(s(x_1,x'_1),\ldots,s(x_k,x'_k))$ would be a Hubie-pol in $\{b\}$ and if Class $(ii)$ were empty we would similarly have a Hubie-pol in $\{a\}$. We will write $k$-tuples with vertical bars to indicate the split between these classes. Suppose there exists a $\overline{z}$ so that $f(a,\ldots,a|b,\ldots,b|\overline{z}) \in \{a,b\}$. Then we can identify all the variables in one among Class $(i)$ or Class $(ii)$ to obtain a new function for which one of these classes is of size one. Note that if, e.g., Class $(i)$ is made singleton, this process may move variables previously in Class $(iii)$ into Class $(ii)$, but never to Class $(i)$. 

Thus we may assume that either Class $(i)$ or Class $(ii)$ is singleton or, for all $\overline{z}$ over $\{a,b\}$, $f(a,\ldots,a|b,\ldots,b|\overline{z}) =c$. Indeed, these singleton cases are dual and thus \mbox{w.l.o.g.} we need only prove one of them. Recall the global assumptions are in force for the remainder of the paper.

\section{Properties of Gap Algebras that are switchable}

\begin{lemma}
\label{lem:fun}
Any algebra over $D$ containing $f$ and $s$ is either collapsible or has binary term operations $p_1$ and $p_2$ so that 
\begin{itemize}
\item $p_1(a,b)=b$ and $p_1(b,a)=p_1(c,a)=c$, \textbf{and}
\item $p_2(a,b)=a$ and $p_2(b,a)=p_2(b,c)=c$.
\end{itemize}
\end{lemma}
\begin{proof}
Consider a tuple $\overline{x}$ over $\{a,b\}$ that witnesses the breaking of $\{a,c\}\{b,c\}$-projectivity for some Class $(i)$ variable from $a$ to $b$; so $f(\overline{x})=b$. Let $\widetilde{x}$ be $\overline{x}$ with the $a$s substituted by $c$ and the $b$s substituted by $a$. If, for each such  $\overline{x}$ over $\{a,b\}$ that witnesses the breaking of $\{a,c\}\{b,c\}$-projectivity for each Class $(i)$ variable, we find $f(\widetilde{x})=a$, then $f(s(x_1,x'_1),\ldots,s(x_k,x'_k))$ is a Hubie-pol in $\{b\}$. Thus, for some such  $\overline{x}$ we find $f(\widetilde{x})=c$. By collapsing the variables according to the division of $\overline{x}$ and $\widetilde{x}$ we obtain a binary function $p_1$ so that $p_1(a,b)=b$ and $p_1(c,a)=c$. We may also see that $p_1(b,a)=c$, since Classes $(i)$ and $(ii)$ are non-empty.

Dually, we consider tuples $\overline{x}$ over $\{a,b\}$ that witnesses the breaking of $\{a,c\}\{b,c\}$-projectivity for Class $(ii)$ variables from $b$ to $a$ to derive a function $p_2$ so that  $p_2(a,b)=a$, $p_2(b,c)=p(b,a)=c$.
\end{proof}

\subsection{The asymmetric case: Class $(i)$ is a singleton and there exists $\overline{z} \in \{a,b\}^*$ so that $f(a|b,\ldots,b|\overline{z})=b$}

We will address the case in which Class $(i)$ is a singleton and there exists $\overline{z} \in \{a,b\}^*$ so that $f(a|b,\ldots,b|\overline{z})=b$ (the like case with Class $(ii)$ being singleton itself being dual).

\begin{proposition}
\label{prop:asymmetric}
Let $f$ be so that Class $(i)$ is a singleton and there exists $\overline{z} \in \{a,b\}^*$ so that $f(a|b,\ldots,b|\overline{z})=b$. Then, either $f$ generates a binary idempotent operation with $ab \mapsto a$ and $ac \mapsto c$, or any algebra on $D$ containing $f$ and $s$ is collapsible.
\end{proposition}
\begin{proof}
Let us consider the general form of $f$,
\[
\begin{array}{c|ccc|ccccc}
a & b & \cdots & b & z_0^0 & \cdots & z_0^{\ell'} & \ & b \\
a & y^1_1 & \cdots & y^{k'}_1 & z_1^1 & \cdots & z_1^{\ell'} & \ & a \\
\vdots & \vdots & \cdots & \vdots & \vdots & \cdots & \vdots & \mapsto & \vdots \\
a & y^1_{m'} & \cdots & y^{k'}_{m'} & z_{m'}^1 & \cdots & z_{m'}^{\ell'} & \ & a \\
\end{array}
\]
where the $y$s and $z$s are from $\{a,b\}$ and we can assume that each $(y^i_1,\ldots,y^i_{m'})$ contains at least one $b$ and also each $(z^i_1,\ldots,z^i_{m'})$ contains at least one $b$. For the latter assumption recall that in Class $(iii)$ we can always find some break of $\alpha\beta$-projectivity from $b$ to $a$. Note that by expanding what we previously called Class $(ii)$ we can build, by possibly identifying variables, a function $f'$ of the form
\[
\begin{array}{c|ccc|ccccc}
a & b & \cdots & b & a & \cdots & a & \ & b \\
a & y^1_1 & \cdots & y^{k}_1 & z_1^1 & \cdots & z_1^{\ell} & \ & a \\
\vdots & \vdots & \cdots & \vdots  & \vdots & \cdots & \vdots & \mapsto & \vdots \\
a & y^1_{m} & \cdots & y^{k}_{m} & z_{m}^1 & \cdots & z_{m}^{\ell} & \ & a \\
\end{array}
\]
where the $y$s and $z$s are from $\{a,b\}$ and we can assume each $(y^i_1,\ldots,y^i_{m})$ contains at least one $b$ and also each $(z^i_1,\ldots,z^i_{m})$ contains a least one $b$. Note that we do not claim the new vertical bars correspond to delineate between Classes $(i)$, $(ii)$ and $(iii)$ under their original definitions, since this is not important to us. We will henceforth assume that $f$ is in the form of $f'$.

Let $x^i_j$ (resp., $v^i_j$) be $a$ if $y^i_j$ (resp., $z^i_j$) is $a$, and be $c$ if $y^i_j$ (resp., $z^i_j$) is $b$. That is, $(x^1_{j}, \ldots,  x^{k}_{j}, v_{j}^1\ \ldots, v_{j}^{\ell})$ is built from $(y^1_{j}, \ldots,  y^{k}_{j}, z_{j}^1\ \ldots, z_{j}^{\ell})$ by substituting $b$s by $c$s. Suppose one of $f(a|x^1_1,\ldots,x^k_1| v_1^1\ \ldots, v_1^{\ell})$, \ldots, $f(a|x^1_m,\ldots,x^k_m | v_{m}^1\ \ldots, v_{m}^{\ell})$ is $c$. Then $f$ generates an idempotent binary operation with $ab \mapsto a$ and $ac \mapsto c$. Thus, we may assume that each of  $f(a|x^1_1,\ldots,x^k_1 | v_1^1\ \ldots, v_1^{\ell})$, \ldots, $f(a|x^1_m,\ldots,x^k_m|$ $v_m^1\ \ldots, v_m^{\ell}))$ is $a$. We now move to consider some cases.

(Case 1: $\ell=0$, \mbox{i.e.} there is nothing to the right of the second vertical bar.) From adversaries of the form $(\{a\}^M)$ and $(\{a,b\}^{m-1},\{b\}^{M-m+1})$ this supports construction of $(\{a,b\}^{m},\{b\}^{M-m})$ and all co-ordinate permutations. We illustrate this with the following diagram which makes some assumptions about the locations of the $b$s in each $(y^i_1,\ldots,y^i_{m})$; nonetheless it should be clear that the method works in general since there is at least one $b$ in $(y^i_1,\ldots,y^i_{m})$.
\[
\begin{array}{c|cccccc}
\{a\} & \{a,b\} & \{a,b\} & \cdots & \{a,b\} & & \{a,b\} \\
\vdots & \vdots & \vdots & \cdots & \vdots & & \{a,b\} \\
\{a\} & \{a,b\} & \{a,b\} & \cdots & \{a,b\} & & \{a,b\} \\
\{a\} & \{b\} & \{a,b\} & \cdots & \{a,b\} & & \{a,b\} \\
\{a\} & \{a,b\}  & \{b\} & \cdots & \{a,b\} & & \{a,b\} \\
\vdots & \vdots & \vdots & \cdots & \vdots & \mapsto & \{a,b\} \\
\{a\} & \{a,b\} & \{a,b\} & \cdots & \{b\} & & \{a,b\} \\
\{a\} & \{b\} & \{b\} & \cdots & \{b\} & & \{b\} \\
\vdots & \vdots & \vdots & \cdots & \vdots & & \{b\} \\
\{a\} & \{b\} & \{b\} & \cdots & \{b\} & & \{b\} \\
\end{array}
\]
Applying $s$ it is clear that the full adversary may be built from, for example, $(D^{m-1},\{a\}^{M-m+1})$ and $(D^{m-1},\{b\}^{M-m+1})$ which demonstrates $(m-1)$-collapsibility.

(Case 2: $\ell\geq 1$.) Here we consider what is $f(a|b,\ldots,b|b,\ldots,b)$. If this is $b$ then we can clearly reduce to the previous case. If it is $a$ then $f(s(x_1,x_1'),\ldots,$ $s(x_{k+\ell+1},x'_{k+\ell+1}))$ is a generalised Hubie-pol in both $aa|bb,\ldots,bb|aa,\ldots,aa$ and $aa|bb,\ldots,bb|bb,\ldots,bb$, and we are collapsible. This is because the composed function on these listed tuples gives $b$ and $a$, respectively, thus permitting to build adversaries of the form $(\{a,b\}^{k+\ell+2},\{a\}^{M-k-\ell-2})$ and  $(\{a,b\}^{k+\ell+2},\{b\}^{M-k-\ell-2})$ from adversaries of the form $(\{a,b\}^{k+\ell+1},\{a\}^{M-k-\ell-1})$ and  $(\{a,b\}^{k+\ell+1},$ $\{b\}^{M-k-\ell-1})$ (\mbox{cf.} Case 1).

Thus, we may assume $f(a|b,\ldots,b|b,\ldots,b)=c$. Using the fact that $f(s(x_1,x_1'),$ $\ldots,s(x_{k+\ell+1},x'_{k+\ell+1}))$ is a generalised Hubie-pol in $aa|bb\ldots bb|aa \ldots aa$ we can build (using $s$ and rather like in Case 1), from adversaries of the form $(\{a\}^{M})$ and  $(D^{(m-1)i},\{b\}^{M-(m-1)i})$, adversaries of the form $(D^{mi},\{b\}^{M-mi})$, and all co-ordinate permutations of this. Similarly, using the fact that $f(s(x_1,x_1'),\ldots,$ $s(x_{k+\ell+1},x'_{k+\ell+1}))$ is a generalised Hubie-pol in $aa|bb\ldots bb|bb \ldots bb$, we can build adversaries of the form $(D^{mi},\{c\}^{M-mi})$. 

(Case 2a: $f(a|c,\ldots,c|c,\ldots,c)=c$.) Consider again
\[
\begin{array}{c|ccc|ccccc}
a & x^1_1 & \cdots & x^{k}_1 & v_1^1 & \cdots & v_1^{\ell} & \ & a \\
\vdots & \vdots & \cdots & \vdots  & \vdots & \cdots & \vdots & \mapsto & \vdots \\
a & x^1_{m} & \cdots & x^{k}_{m} & v_{m}^1 & \cdots & v_{m}^{\ell} & \ & a \\
a & c & \cdots & c & c & \cdots & c & & c\\
\end{array}
\]
where each $(x^i_1,\ldots,x^i_{m})$ and $(v^i_1,\ldots,v^i_{m})$ contains at least one $c$. By amalgamating Classes $(ii)$ and $(iii)$ we obtain some function with the form
\[
\begin{array}{c|ccccc}
a & u^1_1 & \cdots & x^{\nu}_1 & & a \\
\vdots & \cdots & \vdots & & & \vdots \\
a & u^1_m & \cdots & x^{\nu}_m & & a \\
a & c & \cdots & c & & c\\
\end{array}
\]
where each $(u^i_1,\ldots,u^i_{m})$ is in $\{a,c\}^*$ and contains at least one $c$. From adversaries of the form $(D^{r+m-1}, \{c\}^{M-r-m+1})$ and $(\{a\}^{M})$ we can build $(D^{r},\{a,c\}^{M-r})$, and all co-ordinate permutations. We begin, pedagogically preferring to view some $D$s as $\{a,c\}$s,
\[
\begin{array}{c|cccccc}
\{a\} & D & D & \cdots & D & & D \\
\vdots & \vdots & \vdots & \cdots & \vdots & & D \\
\{a\} & D & D & \cdots & D & & D \\
\{a\} & \{c\}  &  \{a,c\} & \cdots &  \{a,c\} & & \{a,c\} \\
\{a\} &  \{a,c\}  & \{c\} & \cdots &  \{a,c\} & & \{a,c\} \\
\vdots & \vdots & \vdots & \cdots & \vdots & \mapsto & \{a,c\} \\
\{a\} &  \{a,c\}  &  \{a,c\} & \cdots & \{c\} & & \{a,c\} \\
\{a\} & \{c\}  & \{c\} & \cdots & \{c\} & & \{c\} \\
\vdots & \vdots & \vdots & \cdots & \vdots & & \{c\} \\
\{a\} & \{c\}  & \{c\} & \cdots & \{c\} & & \{c\} \\
\end{array}
\]
and follow with bottom parts of the form
\[
\begin{array}{c|ccccccccc}
\vdots & \vdots & \vdots & \cdots & \vdots & \mapsto & \{c\} \mbox{ or } \{a,c\} \\
\{a\} & \{a,c\}  & \{a,c\} & \cdots & \{a,c\} & & \{a,c\}. \\
\end{array}
\]
This now supports bootstrapping of the full adversary from adversaries of the form $(D^{m^2},\{a\}^{M-m^2})$, $(D^{m^2},$ $\{b\}^{M-m^2})$ and $(D^{m^2},\{c\}^{M-m^2})$. 

(Case 2b:  $f(a|c,\ldots,c|c,\ldots,c)=a$.) Here, from adversaries of the form $(D^{r+m-1}, \{c\}^{M-r-m+1})$ and $(\{a\}^{M})$ we can directly build $(D^{r+m-1}, \{a\}^{M-r-m+1})$.
\[
\begin{array}{c|cccccc}
\{a\} & D & D & \cdots & D  & & D \\
\vdots & \vdots & \vdots & \cdots & \vdots  & & D \\
\{a\} & D & D & \cdots & D  & & D \\
\{a\} & \{c\}  & \{c\} & \cdots & \{c\} & & \{a\} \\
\vdots & \vdots & \vdots & \cdots & \vdots  & \mapsto & \{a\} \\
\{a\} & \{c\}  & \{c\} & \cdots & \{c\}  & & \{a\} \\
\end{array}
\]
This now supports bootstrapping of the full adversary, similarly as in Case 2a (but slightly simpler).
\end{proof}

Let $\overline{x}:=x_1,\ldots,x_k$ and $\overline{y}:=y_1,\ldots,y_k$ be words over $\{a,b\} \ni x,y$. Let $\wedge(x,y)=a$ if $a \in \{x,y\}$ and $b$ otherwise. Let $\vee(x,y)=b$ if $b \in \{x,y\}$ and $a$ otherwise. This corresponds with considering $a$ as $\bot$ and $b$ as $\top$. Define $\wedge(\overline{x},\overline{y}):=(\wedge(x_1,y_1),\ldots,\wedge(x_k,y_k))$ and $\vee(\overline{x},\overline{y}):=(\vee(x_1,y_1),\ldots,\vee(x_k,y_k))$.  We are most interested in words 
\begin{itemize}
\item[A] $(\overline{x} | a,\ldots,a | \overline{z})$, such that $f(\overline{x} | a,\ldots,a | \overline{z})=a$, and for no $\overline{x}'\neq \overline{x}$ and $\overline{z}'$ over $\{a,b\}$ do we have  $(\overline{x}' | a,\ldots,a | \overline{z}')$ with $\vee(\overline{x},\overline{x}')=\overline{x}'$ so that $f(\overline{x}' | a,\ldots,a | \overline{z}')=a$.
\item[B] $(b,\ldots,b | \overline{y} | \overline{z})$, such that $f(b,\ldots,b | \overline{y} | \overline{z})=b$, and for no $\overline{y}'\neq \overline{y}$  and $\overline{z}'$ over $\{a,b\}$ do we have $(b,\ldots,b | \overline{y}' | \overline{z}')$ with $\wedge(\overline{y},\overline{y}')=\overline{y}'$ so that $f(b,\ldots,b | \overline{y}' | \overline{z}')=b$.
\end{itemize}
Such $\overline{x}$ and $\overline{y}$ are in a certain sense \emph{maximal}, but the sense of maximality is dual in Case B from Case A. $\overline{x}$ is maximal under inclusion for the number of $b$s it contains and $\overline{y}$ is maximal under inclusion for the number of $a$s it contains. In the asymmetric case that we consider here w.l.o.g., only Case A above will be salient, but we introduce both now for pedagogical reasons.

\begin{lemma}
\label{lem:r4-asymmetric}
Let $f$ be so that Class $(i)$ is a singleton and there exists $\overline{z} \in \{a,b\}^*$ so that $f(a|b,\ldots,b|\overline{z})=b$. Then any algebra over $D$ containing $f$ and $s$ is either collapsible or has a $4$-ary term operation $r_4$ so that
\[
\begin{array}{ccc}
abab & r_4 & a \\
abba & \rightarrow & a \\
abbb & & c \\
\end{array}
\]
\end{lemma}
\begin{proof}
Recall $\exists \overline{z}$ so that $f(a|b,\ldots,b|\overline{z})=b$. Note that if exists $\overline{z}'$ over $\{a,b\}$ so that $f(a|b,\ldots,b|\overline{z}')=a$ then we have that $f(s(v_1,v'_1),\ldots,s(v_{k+\ell+1},v'_{k+\ell+1}))$ is a generalised Hubie-pol in both $bb\ldots bb|aa \ldots aa|\widehat{z}$ and $bb\ldots bb|aa \ldots aa|\widehat{z}'$, where we build widehat from overline by doubling each entry where it sits, and we become collapsible. It therefore follows that there must exist distinct $\overline{y}_1$, $\overline{y}_2$, $\overline{z}_1$ and $\overline{z}_2$ (all over $\{a,b\}$) so that $f(a|\overline{y}_1|\overline{z}_1)=a$, $f(a|\overline{y}_2|\overline{z}_2)=a$ but $f(a|\vee(\overline{y}_1,\overline{y}_2)|\overline{z}_1)\neq a$. By collapsing co-ordinates we get $f'$ so that
\[
\begin{array}{ccc}
aabb & f' & a \\
abab &  \rightarrow & a \\
abbb & & \mbox{$a$ or $c$} \\
\end{array}
\]
The result follows by permuting co-ordinates, possibly in new combination through $s$ and the second co-ordinate.
\end{proof}

\subsection{The symmetric case: for every $\overline{z} \in \{a,b\}^*$ we have $f(a,\ldots,a|b,\ldots,b|\overline{z})=c$}

\begin{proposition}
\label{prop:symmetric}
Let $f$ be so that neither Class $(i)$ nor Class $(ii)$ is a singleton and so that for every $\overline{z} \in \{a,b\}^*$ we have $f(a,\ldots,a|b,\ldots,b|\overline{z})=c$. Then, either $f$ generates a binary idempotent operation with $ab \mapsto a$ and $ac \mapsto c$ or a binary idempotent operation with $ab \mapsto b$ and $cb \mapsto c$, or any algebra on $D$ containing $f$ and $s$ is collapsible.
\end{proposition}
\begin{proof}
Let us consider the general form of $f$,
\[
\begin{array}{ccc|ccc|ccccc}
x^1_1 & \cdots & x^k_1 & b & \cdots & b & w_1^1 & \cdots & w_1^\ell & \ & b \\
\vdots & \vdots & \vdots & \vdots & \vdots & \vdots & \vdots & \vdots & \vdots & \mapsto & \vdots \\
x^1_m & \cdots & x^k_m & b & \cdots & b & w_m^1 & \cdots & w_m^\ell & \ & b \\
\\
a & \cdots & a & y^1_1 & \cdots & y^\kappa_1 & z_1^1 & \cdots & z_1^\ell & \ & a \\
\vdots & \vdots & \vdots & \vdots & \vdots & \vdots & \vdots & \vdots & \vdots & \mapsto & \vdots \\
a & \cdots & a & y^1_\mu & \cdots & y^\kappa_\mu & z_\mu^1 & \cdots & z_\mu^\ell & \ & a \\
\end{array}
\]
where the $x$s, $y$s, $z$s and $w$s are from $\{a,b\}$ and we can assume that each $(x^i_1,\ldots,x^i_m)$ and $(w^i_1,\ldots,w^i_m)$ contain at least one $a$ and $(y^i_1,\ldots,y^i_\mu)$ and $(z^i_1,\ldots,$ $z^i_\mu)$ contains at least one $b$. As in the previous proof we can make an assumption that each $(x^1_i, \cdots, x^k_i , b, \ldots , b , w_i^1 , \ldots , w_i^\ell)$, with $a$ substituted for $c$, still maps under $f$ to $b$. Similarly,  each $(a,\ldots,a,y^1_i, \cdots, y^\kappa_i, z_i^1 , \ldots , z_i^\ell)$, with $b$ substituted for $c$, still maps under $f$ to $a$.

Since for each  $\overline{z} \in \{a,b\}^*$ we have $f(a,\ldots,a|b,\ldots,b|\overline{z})=c$ we can deduce that from the adversaries $(D^{(k+\kappa+\ell-1)i},\{a\}^{M-(k+\kappa+\ell-1)i})$ and  $(D^{(k+\kappa+\ell-1)i},$ $\{b\}^{M-(k+\kappa+\ell-1)i})$, adversaries of the form $(D^{(k+\kappa+\ell)i},\{c\}^{M-(k+\kappa+\ell)i})$, and all co-ordinate permutations of this. 

We now make some case distinctions based on whether $f(a,\ldots,a|c,\ldots,c|$ $c,\ldots,c)=c$ or $a$ and $f(c,\ldots,c|b,\ldots,b|c,\ldots,c)=c$ or $b$ (note that possibly Class $(iii)$ is empty). However, the method for building the full adversary from certain collapsings proceeds very similarly to Cases 2a and 2b from Proposition~\ref{prop:symmetric}. We give an example below as to how, in the case $f(a,\ldots,a|c,\ldots,c|c,\ldots,c)=c$, we mimic Case 2a from Proposition~\ref{prop:asymmetric} to derive a function from this that builds, from adversaries of the form $(D^{r+m-1}, \{a\}^{M-(r+m-1)})$ and $(D^{r+2m-1},\{c\}^{M-(r+m-1)})$, we can build $(D^{r+m},\{a,c\}^{M-m-r})$. For pedagogic reasons we prefer to view some $D$s as $\{a,c\}$s,
\[
\begin{array}{cccc|cccccc}
\{a\} & D & \cdots & D & D & D & \cdots & D & & D \\
D & \{a\} & \cdots & D & D & D & \cdots & D & & D \\
\vdots & \vdots & \cdots & \vdots & \vdots & \vdots & \cdots & \vdots  & & D \\
D & D & \cdots & \{a\} & D & \cdots & D & D & & D \\
\{a\} & \{a\} & \cdots & \{a\}  & \{c\}  & \{a,c\} & \cdots & \{a,c\} & & \{a,c\} \\
\{a\} & \{a\} & \cdots & \{a\} & \{a,c\}  & \{c\} & \cdots & \{a,c\} & & \{a,c\} \\
\vdots & \vdots & \vdots & \vdots & \vdots & \vdots & \cdots & \vdots  & \mapsto & \{a,c\} \\
\{a\} & \{a\} & \cdots & \{a\} & \{a,c\} & \{a,c\} & \cdots & \{c\} & & \{a,c\} \\
\{a\} & \{a\} & \cdots & \{a\} &\{c\}  & \{c\} & \cdots & \{c\} & & \{c\} \\
\vdots & \vdots & \vdots & \vdots & \vdots & \vdots & \cdots & \vdots & & \{c\} \\
\{a\} & \{a\} & \cdots & \{a\} & \{c\}  & \{c\} & \cdots & \{c\} & & \{c\} \\
\end{array}
\]
\end{proof}

\begin{lemma}
\label{lem:r4-symmetric}
Let $f$ be so that neither Class $(i)$ nor Class $(ii)$ is a singleton and so that for every $\overline{z} \in \{a,b\}^*$ we have $f(a|b,\ldots,b|\overline{z})=c$. Any algebra over $D$ containing $f$ is either collapsible or contains a $4$-ary operations $r^a_4$ and $r^b_4$ with properties
\[
\begin{array}{ccc}
\begin{array}{ccc}
abab & r^a_4 & a \\
abba & \rightarrow & a \\
abbb & & c \\
\end{array}
& \ \ \ \ \ \  \mbox{\textbf{and}} \ \ \ \ \ \ &
\begin{array}{ccc}
abab & r^b_4 & b \\
abba & \rightarrow & b \\
abaa & & c \\
\end{array}
\end{array}
\]
\end{lemma}
\begin{proof}
The proof proceeds exactly as in Lemma~\ref{lem:r4-asymmetric}.
\end{proof}

An important special case  of the previous lemma, which is satisfied by Chen's $(\{a,b,c\};r,s)$ is as follows.

\vspace{0.2cm}
\noindent \textbf{Zhuk Condition}. $\mathbb{A}$ has idempotent term operations, binary $p$ and ternary operation $r_3$, so that either
\[
\begin{array}{c}
\left(
  \begin{array}{ccc}
    \begin{array}{ccc}
    aab & r_3 & a \\
    aba & \rightarrow & a \\
    abb & & c \\
    \end{array}
  & \ \ \ \ \ \  \mbox{\textbf{and}} \ \ \ \ \ \ &
    \begin{array}{ccc}
    ab & p & a \\
    ac & \rightarrow & c \\
    \end{array}
  \end{array}
\right)
\\
\mbox{\textbf{or}} \\
\left(
\begin{array}{ccc}
   \begin{array}{ccc}
   bab & r_3 & b \\
   bba & \rightarrow & b \\
   baa & & c \\
   \end{array}
& \ \ \ \ \ \  \mbox{\textbf{and}} \ \ \ \ \ \ &
    \begin{array}{ccc}
    ab & p & b \\
    cb & \rightarrow & c \\
    \end{array}
  \end{array}
\right)
\end{array}
\]

\section{About essential relations}

We assume that all relations are defined on the finite set $\{a,b,c\}$.
A relation $\rho$ is called \textit{essential} if
it cannot be represented as a conjunction of relations with smaller arities.
A tuple $(a_{1},a_{2},\ldots,a_{n})$ is called \textit{essential for a relation $\rho$}
if $(a_{1},a_{2},\ldots,a_{n})\notin\rho$ and
for every $i\in \{1,2,\ldots,n\}$ there exists $b\in A$ such that 
$(a_{1},\ldots,a_{i-1},b,a_{i+1},\ldots,a_{n})\in\rho.$ Let us define a relation $\tilde{\rho}$ for every relation $\rho \subseteq D^n$. Put $\sigma_i(x_1,\ldots,x_{i-1},x_{i+1},\ldots,x_{n}) := \exists y \ \rho(x_1,\ldots,x_i,y,x_{i+1},\ldots,x_n)$
and let 
\[ \tilde{\rho}(x_1,\ldots,x_n) := \sigma_1(x_2,x_3,\ldots,x_n) \wedge  \sigma_2(x_1,x_3,\ldots,x_n) \wedge \ldots \wedge  \sigma_1(x_1,x_2,\ldots,x_{n-1}). \]
\begin{lemma}\label{sushnabor}
A relation $\rho$ is essential iff there exists an essential tuple for $\rho$.
\end{lemma}
\begin{proof}
(Forwards.) By contraposition, if $\rho$ is not essential, then $\tilde{\rho}$ is equivalent to $\rho$, and there can not be an essential tuple.

(Backwards.) An essential tuple witnesses that a relation is essential.
\end{proof}
\begin{lemma}
\label{lem:Dmitriy-micro}
Suppose $(c,c,x_3,\ldots,x_n)$ is an essential tuple for $\rho$. Then $\rho$ is not preserved by $s$.
\end{lemma}
\begin{proof}
Since $(c,c,x_3,\ldots,x_n)$ is an essential tuple, $(x_1,c,x_3,\ldots,x_n)$ and $(c,x_2,x_3,$ $\ldots,x_n)$ are in $\rho$ for some $x_1$ and $x_2$. But applying $s$ now gives the contradiction.
\end{proof}
For a tuple $\mathbf{y}$, we denote its $i$th co-ordinate by $\mathbf{y}(i)$. For $n\geq 3$, we define the arity $n+1$ idempotent operation $f^a_n$ as follows
\[
\begin{array}{c}
f^a_n(a,a,\ldots,a,a)=a \\
f^a_n(b,b,\ldots,b,b)=b \\
f^a_n(b,a,\ldots,a,a)=a \\
f^a_n(a,b,\ldots,a,a)=a \\
\vdots \\
f^a_n(a,a,\ldots,b,a)=a \\
f^a_n(a,a,\ldots,a,b)=a \\
\mbox{else $c$}
\end{array}
\]
We define $f^b_n$ similarly with $a$ and $b$ swapped. These functions are very similar to partial near-unanimity functions.
\begin{lemma}
\label{lem:Dmitriy}
Suppose $\mathbb{A}$ is a Gap Algebra, that is not $\alpha\beta$-projective, so that $\mathbb{A}$ satisfies the Zhuk Condition. Then either 
\begin{itemize}
\item any relation $\rho \in \mathrm{Inv}(\mathbb{A})$ of arity $h<n+1$ is preserved by $f^a_n$, \textbf{or} 
\item any relation $\rho \in \mathrm{Inv}(\mathbb{A})$ of arity $h<n+1$ is preserved by $f^b_n$.
\end{itemize}
\end{lemma}
\begin{proof}
Suppose \mbox{w.l.o.g.} that the Zhuk Condition is in the first regime and has idempotent term operations, binary $p$ and ternary operation $r_3$, so that
\[
  \begin{array}{ccc}
    \begin{array}{ccc}
    aab & r_3 & a \\
    aba & \rightarrow & a \\
    abb & & c \\
    \end{array}
  & \ \ \ \ \ \  \mbox{\textbf{and}} \ \ \ \ \ \ &
    \begin{array}{ccc}
    ab & p & a \\
    ac & \rightarrow & c \\
    \end{array}
  \end{array}
\]
We prove this statement for a fixed $n$ by induction on $h$. For $h = 1$ we just need to
check that $f_n:=f^a_n$ preserves the unary relations $\{a, c\}$ and $\{b, c\}$.

Assume that $\rho$ is not preserved by $f_n$, then there exist tuples $\mathbf{y}_1,\ldots,\mathbf{y}_{n+1} \in \rho$ such that $f_n(\mathbf{y}_1,\ldots,\mathbf{y}_{n+1})=\gamma \notin \rho$. We consider a matrix whose columns are $\mathbf{y}_1,\ldots,\mathbf{y}_{n+1}$. Let the rows of this matrix be $\mathbf{x}_1,\ldots,\mathbf{x}_h$.

By the inductive assumption every $\sigma_i$ from the definition of $\widetilde{\rho}$ is preserved by $f_n$, which means that $\widetilde{\rho}$ is preserved by $f_n$, which means that $\gamma \notin \rho$ and $\gamma$ is an essential tuple for $\rho$.

We consider two cases. First, assume that $\gamma$ doesn't contain $c$. Then it follows from the definition that every $\mathbf{x}_i$ contains at most one element that differs from $\gamma(i)$. Since $n+1>h$, there exists $i \in \{1, 2, \ldots , n + 1\}$ such that $\mathbf{y}_i = \gamma$. This contradicts the fact that $\gamma \notin \rho$.

Second, assume that $\gamma$ contains $c$. Then by Lemma~\ref{lem:Dmitriy-micro}, $\gamma$ contains exactly one $c$. \mbox{W.l.o.g.} we assume that $\gamma(1) = c$. It follows from the definition of $f_n$ that $\mathbf{x}_i$ contains at most one element that differs from $\gamma(i)$ for every $i \in \{2, 3, \ldots , h\}$. Hence, since $n+1>h$, for some $k \in \{1, 2, \ldots , n+ 1\}$ we have $\mathbf{y}_k(i) = \gamma(i)$ for every $i \in \{2, 3, \ldots , h\}$. Since $f_n(\mathbf{x}_1) = c$, we have one of three subcases. First subcase, $\mathbf{x}_1(j) = c$ for some $j$. We need one of the properties
\[
\begin{array}{cc|c}
\mathbf{y}_k & \mathbf{y}_j & \gamma \\
\hline
a & c & c \\
a & b & a \\
\end{array}
\mbox{ \ \ \ \ \ \ \ \ \ \ \ \ \ \ \ \ \ \ }
\begin{array}{cc|c}
\mathbf{y}_k & \mathbf{y}_j & \gamma \\
\hline
b & c & c \\
a & b & a \\
\end{array}
\]
and we can see that the functions from Lemma~\ref{lem:fun} or the definition of the Zhuk Condition suffice, which contradicts our assumptions.

Second subcase, $\mathbf{y}_k(1) = b, \mathbf{y}_m(1) = a$ for some $m \in \{1, 2, \ldots , n + 1\}$. We need the property 
\[
\begin{array}{cc|c}
\mathbf{y}_k & \mathbf{y}_m & \gamma \\
\hline
b & a & c \\
a & b & a \\
\end{array}
\]
can check that a function from Lemma~\ref{lem:fun} suffices, which contradicts our assumptions.

Third subcase, $\mathbf{y}_k(1) = a, \mathbf{y}_m(1) = b$ and $\mathbf{y}_l(1) = b$ for $m, l \in \{1, 2, \ldots , n + 1\}\setminus \{k\}$, $m \neq l$. We need the property
\[
\begin{array}{ccc|c}
\mathbf{y}_k & \mathbf{y}_m & \mathbf{y}_l & \gamma \\
\hline
a & b & b & c \\
a & a & b & a \\
a & b & a & a \\
\end{array}
\]
and we can check that the $r_3$ from the Zhuk Condition suffices, which contradicts our assumptions. This completes the proof.
\end{proof}
\begin{corollary}
\label{cor:Dmitriy}
Suppose $\mathbb{A}$ is a Gap Algebra, that is not $\alpha\beta$-projective so that $\mathbb{A}$ satisfies the Zhuk Condition. Then, for every finite subset of $\Delta$ of Inv$(\mathbb{A})$, Pol$(\Delta)$ is collapsible.
\end{corollary}
\begin{proof}
$f^a_n$ is a Hubie-pol in $\{b\}$ and $f^b_n$ is a Hubie-pol in $\{a\}$.
\end{proof}
For $n\geq 2$, we define the arity $n+2$ idempotent operation $\widehat{f}^a_n$ as follows
\[
\begin{array}{c}
f^a_n(a,a,a,\ldots,a,a)=a \\
f^a_n(b,b,b,\ldots,b,b)=b \\
f^a_n(b,b,a,\ldots,a,a)=a \\
f^a_n(b,a,b,\ldots,a,a)=a \\
\vdots \\
f^a_n(b,a,a,\ldots,b,a)=a \\
f^a_n(b,a,a,\ldots,a,b)=a \\
\mbox{else $c$}
\end{array}
\]
We define $\widehat{f}^b_n$ similarly with $a$ and $b$ swapped.
\begin{lemma}
\label{lem:Dmitriy-long}
Suppose $\mathbb{A}$ is a Gap Algebra that is not $\alpha\beta$-projective. Then either 
\begin{itemize}
\item any relation $\rho \in \mathrm{Inv}(\mathbb{A})$ of arity $h<n+2$ is preserved by $\widehat{f}^a_n$, \textbf{or} 
\item any relation $\rho \in \mathrm{Inv}(\mathbb{A})$ of arity $h<n+2$ is preserved by $\widehat{f}^b_n$.
\end{itemize}
\end{lemma}
\begin{proof}
Suppose \mbox{w.l.o.g.} that we are either in the asymmetric case with Class $(i)$ singleton and  there exists $\overline{z} \in \{a,b\}^*$ so that $f(a|b,\ldots,b|\overline{z})=b$ \textbf{or} we are in the symmetric case and we have an idempotent term operation $p$ mapping $ab \mapsto a$ and $ac \mapsto c$.

We prove this statement for a fixed $n$ by induction on $h$. For $h = 1$ we just need to
check that $\widehat{f}_n:=\widehat{f}^a_n$ preserves the unary relations $\{a, c\}$ and $\{b, c\}$.

Assume that $\rho$ is not preserved by $f_n$, then there exist tuples $\mathbf{y}_1,\ldots,\mathbf{y}_{n+2} \in \rho$ such that $\widehat{f}_n(\mathbf{y}_1,\ldots,\mathbf{y}_{n+2})=\gamma \notin \rho$. We consider a matrix whose columns are $\mathbf{y}_1,\ldots,\mathbf{y}_{n+2}$. Let the rows of this matrix be $\mathbf{x}_1,\ldots,\mathbf{x}_h$.

By the inductive assumption every $\sigma_i$ from the definition of $\widetilde{\rho}$ is preserved by $\widehat{f}_n$, which means that $\widetilde{\rho}$ is preserved by $\widehat{f}_n$, which means that $\gamma \notin \rho$ and $\gamma$ is an essential tuple for $\rho$.

We consider two cases. First, assume that $\gamma$ doesn’t contain $c$. Then it follows from the definition that every $\mathbf{x}_i$ contains at most one element that differs from $\gamma(i)$. Since $n+2>h$, there exists $i \in \{1, 2, \ldots , n + 1\}$ such that $\mathbf{y}_i = \gamma$. This contradicts the fact that $\gamma \notin \rho$.

Second, assume that $\gamma$ contains $c$. Then by Lemma~\ref{lem:Dmitriy-micro}, $\gamma$ contains exactly one $c$. \mbox{W.l.o.g.} we assume that $\gamma(1) = c$. It follows from the definition of $\widehat{f}_n$ that $\mathbf{x}_i$ contains at most one element that differs from $\gamma(i)$ for every $i \in \{2, 3, \ldots , h\}$. Hence, since $n+2>h$, for some $k \in \{2, \ldots , n+2\}$ we have $\mathbf{y}_k(i) = \gamma(i)$ for every $i \in \{2, 3, \ldots , h\}$. Since $\widehat{f}_n(\mathbf{x}_1) = c$, we have one of four subcases. First subcase, $\mathbf{x}_1(j) = c$ for some $j$. We need one of the properties
\[
\begin{array}{cc|c}
\mathbf{y}_k & \mathbf{y}_j & \gamma \\
\hline
a & c & c \\
a & b & a \\
\end{array}
\mbox{ \ \ \ \ \ \ \ \ \ \ \ \ \ \ \ \ \ \ }
\begin{array}{cc|c}
\mathbf{y}_k & \mathbf{y}_j & \gamma \\
\hline
b & c & c \\
a & b & a \\
\end{array}
\]
and we can see that the functions from Lemma~\ref{lem:fun}, or Proposition~\ref{prop:asymmetric} or Proposition~\ref{prop:symmetric}, suffice which contradicts our assumptions.

Second subcase, $\mathbf{y}_k(1) = b, \mathbf{y}_m(1) = a$ for some $m \in \{1, 2, \ldots , n + 1\}$. We need the property 
\[
\begin{array}{cc|c}
\mathbf{y}_k & \mathbf{y}_m & \gamma \\
\hline
b & a & c \\
a & b & a \\
\end{array}
\]
can check that a function from Lemma~\ref{lem:fun} suffices, which contradicts our assumptions.

For Case 3, $\mathbf{y}_k(1) = a, \mathbf{y}_m(1) = b$ and $\mathbf{y}_l(1) = b$ for some $m, l \in \{1, 2, \ldots , n + 1\}\setminus \{k\}$, $m \neq l$ (possibly $1 \in \{m,l\}$). We now split into two subsubcases: either $\mathbf{y}_1(1)=b$ and we need the property
\[
\begin{array}{cccc|c}
\mathbf{y}_1 &  \mathbf{y}_k & \mathbf{y}_m & \mathbf{y}_l & \gamma \\
\hline
b & a & b & b & c \\
b & a & a & b & a \\
b & a & b & a & a. \\
\end{array}
\]
Here we can check that $r_4$, from Proposition~\ref{prop:asymmetric} or Proposition~\ref{prop:symmetric}, with co-ordinates $1$ and $2$ permuted, suffices, which contradicts our assumptions. Or we have  $\mathbf{y}_1(1)=a$ and we need the property
\[
\begin{array}{cccc|c}
\mathbf{y}_1 & \mathbf{y}_k & \mathbf{y}_m & \mathbf{y}_l & \gamma \\
\hline
a & a & b & b & c \\
b & a & a & b & a \\
b & a & b & a & a. \\
\end{array}
\]
For this $p(x_2,(p_1(x_4,p_1(x_2,x_1))))$ suffices where $p_1$ comes from Lemma~\ref{lem:fun} and $p$ is as before in this proof (\mbox{cf.} Proposition~\ref{prop:asymmetric} and Proposition~\ref{prop:symmetric}).
\[
\begin{array}{cccc|c|c|c}
x_1 & x_2 & x_3 & x_4 & p_1(x_2,x_1) & p_1(x_4,p_1(x_2,x_1)) & p(x_2,(p_1(x_4,p_1(x_2,x_1)))) \\
\hline
a & a & b & b & a & c & c \\
b & a & a & b & b & b & a\\
b & a & b & a & b & b & a \\
\end{array}
\]
This completes the proof.
\end{proof}
\begin{corollary}
\label{cor:Dmitriy-long}
Suppose $\mathbb{A}$ is a Gap Algebra that is not $\alpha\beta$-projective. Then, for every finite subset of $\Delta$ of Inv$(\mathbb{A})$, Pol$(\Delta)$ is collapsible.
\end{corollary}
\begin{proof}
$\widehat{f}^a_n$ is a Hubie-pol in $\{b\}$ and $\widehat{f}^b_n$ is a Hubie-pol in $\{a\}$.
\end{proof}

\section{Collapsibility}
Let $t$ be the $4$-ary operation that maps
\[
\begin{array}{ccc}
abab & & b \\
abba &  t & a \\
cbbc & \mapsto & b \\
acac & & a \\
\mathit{else} & & c. 
\end{array}
\]
\begin{lemma}
The algebra $(D;s,t)$ is $7$-collapsible from source $\{a,b\}$.
\end{lemma}
\begin{proof}
$t(s(x_1,x'_1),\ldots,s(x_4,x'_4))$ is a generalised Hubie-pol in $aabbaabb \mapsto b$ and $aabbbbaa \mapsto a$.
\end{proof}

\begin{lemma}
Let $g$ be a $m$-ary term operation of the algebra $(D;s,t)$. There exists co-ordinates $i,j\in [m]$ so that 
\begin{itemize}
\item $g(D,\ldots,D,\{a,c\},D,\ldots,D) \not\ni b$, where $\{a,c\}$ is in the $i$th position, and 
\item $g(D,\ldots,D,\{b,c\},D,\ldots,D) \not\ni a$, where $\{b,c\}$ is in the $j$th position.
\end{itemize}
\end{lemma}
\begin{proof}
The proof is by induction on the term complexity of $g$. We will prove the first case as the second is dual.

(Base case.) For $s$ one may choose either co-ordinate, for $t$ one can choose the second.

(Inductive step.) If $g$ is of the form $s(f_1(\ldots),f_2(\ldots))$, and $f_1$ satisfies the inductive hypothesis at co-ordinate $i$, then we may take this co-ordinate again (now seen as an argument of $g$). If $g$ is of the form $t(f_1(\ldots),f_2(\ldots),f_3(\ldots),f_4(\ldots))$, and $f_2$ satisfies the inductive hypothesis as co-ordinate $i$, then we may take this co-ordinate again (now seen as an argument of $g$).
\end{proof}

\begin{corollary}
The algebra $(D;s,t)$ is collapsible from neither source $\{a\}$ nor source $\{b\}$.
\end{corollary}

\begin{lemma}
Let $g$ be a $m$-ary term operation of the algebra $(D;s,t)$. There exists co-ordinate $i\in [m]$ so that 
\begin{itemize}
\item $g(D,\ldots,D,\{c\},D,\ldots,D) =\{c\}$, where $\{c\}$ is in the $i$th position.
\end{itemize}
\end{lemma}
\begin{proof}
The proof is by induction on the term complexity of $g$.

(Base case.) For $s$ one may choose either co-ordinate, for $t$ one can choose the third.

(Inductive step.) If $g$ is of the form $s(f_1(\ldots),f_2(\ldots))$, and $f_1$ satisfies the inductive hypothesis as co-ordinate $i$, then we may take this co-ordinate again (now seen as an argument of $g$). If $g$ is of the form $t(f_1(\ldots),f_2(\ldots),f_3(\ldots),f_4(\ldots))$, and $f_3$ satisfies the inductive hypothesis as co-ordinate $i$, then we may take this co-ordinate again (now seen as an argument of $g$).
\end{proof}

\begin{corollary}
The algebra $(D;s,t)$ is not collapsible from source $\{c\}$.
\end{corollary}

\bibliographystyle{acm}

\end{document}